\documentclass[11pt]{amsart}
\usepackage[letterpaper,margin=1in]{geometry}               
\usepackage{amssymb,amsmath,amsfonts,stmaryrd}
\usepackage{xcolor}
\usepackage[colorlinks,
             linkcolor=blue,
             citecolor=green!70!black,
             pdfproducer={pdfLaTeX},
             pdfpagemode=None,
             bookmarksopen=true
             bookmarksnumbered=true]{hyperref}
\usepackage{graphicx}
\usepackage{multicol}
\usepackage{tikz-cd}
\usepackage{lipsum}
\usepackage{adjustbox}
\usepackage{tikz}
\usepackage{bbold}
\usepackage{tikz-cd}
\usetikzlibrary{shapes.geometric,decorations.markings,arrows,decorations.pathreplacing,calc}
\usepackage{arydshln}

\newcommand\arXiv[1]{\href{http://arxiv.org/abs/#1}{\nolinkurl{arXiv:#1}}}
\newcommand\MRnumber[1]{\href{http://www.ams.org/mathscinet-getitem?mr=#1}{\nolinkurl{MR#1}}}
\newcommand\DOI[1]{\href{http://dx.doi.org/#1}{\nolinkurl{DOI:#1}}}
\newcommand\MAILTO[1]{\href{mailto:#1}{\nolinkurl{#1}}}

\newcounter{mainthm}

\newtheorem{dummy}{Dummy}[section]

\newtheorem{lemma}[dummy]{Lemma}
\newtheorem{proposition}[dummy]{Proposition}

\newtheorem{definition}[dummy]{Definition}

\theoremstyle{definition}
\newtheorem*{rem}{Remark}

\newtheorem*{example}{Example}

\usepackage{dsfont}
\renewcommand\mathbb\mathds

\newcommand\bC{\mathbb C}

\newcommand\bT{\mathbb T}

\newcommand\bZ{\mathbb Z}

\newcommand\cE{\mathcal E}

\newcommand\cH{\mathcal H}

\newcommand\cM{\mathcal M}

\newcommand\cT{\mathcal T}

\newcommand\cZ{\mathcal Z}

\newcommand\rB{\mathrm B}

\newcommand\rH{\mathrm H}

\newcommand\rL{\mathrm L}

\newcommand\rO{\mathrm O}

\newcommand\rR{\mathrm R}

\newcommand\rU{\mathrm U}

\newcommand\SO{\mathrm {SO}}
\newcommand\SU{\mathrm {SU}}

\newcommand\SL{\mathrm {SL}}
\newcommand\Dih{\mathrm {Dih}}

\usepackage[mathscr]{euscript}

\newcommand\longto\longrightarrow
\newcommand\mono\hookrightarrow
\newcommand\epi\twoheadrightarrow

\newcommand\<\langle
\renewcommand\>\rangle
\newcommand\sminus\smallsetminus

\DeclareMathOperator{\Sym}{Sym}
\DeclareMathOperator{\Spin}{Spin}
\DeclareMathOperator{\Pin}{Pin}
\DeclareMathOperator{\Dic}{Dic}

\DeclareMathOperator{\Tr}{Tr}

\newcommand\define[1]{\emph{#1}}

\title{ Genus-One data and Anomaly Detection} 
\author[Yu]{ Matthew Yu$^{1}$}
\thanks{It is a pleasure to thank Theo Johnson-Freyd for proposing this question, engaging in many useful discussions, and giving comments on the draft.
The author would also like to thank Justin Kulp for sharing his knowledge regarding 2d CFTs,
for many helpful discussions, and for giving comments on the draft. Research at the Perimeter Institute is supported by the Government of Canada through Industry Canada and by the Province of Ontario through the Ministry of Economic Development and Innovation. 
 The Perimeter Institute is in the Haldimand Tract, land promised to the Six Nations.
 \\[6pt]
$^1$ \textsc{Perimeter Institute for Theoretical Physics, Waterloo, Ontario}.
\\[6pt]
 \MAILTO{myu@perimeterinstitute.ca}
}

\begin{document}
\begin{abstract}
We introduce the notion of genus-one data for theories in (1+1)-dimensions with an anomalous finite group global symmetry. We outline the groups for which genus-one data is effective in  detecting the anomaly, and also show that genus-one data is insufficient to detect the anomaly for dicyclic groups. 
\end{abstract}
\maketitle

\section{Introduction}
Theories in $d$ spacetime dimensions with a global symmetry group $G$ can have obstructions to promoting the global symmetry to a gauge symmetry. In field theory, one way to work with a global symmetry is to couple it to a background gauge field.  Promoting the symmetry to a gauge symmetry is the same as asking whether it is possible to integrate over these background fields in the path integral, a process known as \textit{gauging} or in other contexts \textit{orbifolding} \cite{gaiotto2020orbifold}.  When gauging is not possible for a certain symmetry, then we say that the theory has an  't Hooft anomaly, that is, an  obstruction classified by a class in $\rH^{d+1}(G \,; \rU(1))$ when the dimension is low.  This means that anomalies are inherently topological in nature, and are moreover robust to deformations by local operators.  These deformations may flow the theory to be in a strongly coupled regime, which makes the dynamics hard to discern.  Information about the anomalies puts constraints on the dynamics, enough so that we are able to make conjectures about the strongly coupled phases.  The anomaly is always present along the renormalization group flow, so whatever value the anomaly takes in, say, a weakly coupled regime, must be matched in the strongly coupled regime. 

If the symmetry has an anomaly, then detecting the anomaly, i.e. determining what value the anomaly takes is often not a very systematic process and depends on the symmetry at hand.  One such way of detecting the anomaly is to study the Hilbert space of the theory on some manifold, such as the torus \cite{Delmastro:2021xox}.  
However, there is no guarantee that one can detect all such anomalies for any symmetry simply by applying one particular method.  It was shown in \cite{Lin:2021udi} how to detect anomalies of $\bZ_N$ global symmetry in (1+1)d unitary conformal field theory, but not much attention has been given to anomalies of nonabelian global symmetry.
We will be interested in a method of detecting anomalies in (1+1)d theories by constructing a stack $\mathcal{M}^G$. This stack will contain information of the theory when placed on a torus with a $G$-bundle, for $G$ a finite group. In the full construction of $\cM^G$ we will have to quotient by automorphisms of the torus and trivializations of the $G$-bundle.  We refer to this stack and automorphism information as \textit{genus-one} data.
Over each point of this stack is a torus bundle, and by integrating the anomaly $\alpha \in \rH^3(\rB G\,; \rU(1))$ over the torus bundle, we see that $\cM^G$ furnishes a line bundle.  The kernel of this integration map is precisely the failure to detect $\alpha$. 

We remark in passing an application of this line bundle.  In problems involving Moonshine there is a connection between ``analytic" data involving modularity and growth rates of certain functions, with representations of finite groups.  The modularity is particularly important because this combines with the finite groups into holomorphic sections of a line bundle on $\cM^G$ \cite{Cheng:2016nto}. This line bundle is the integral of $\alpha$ (in all computed examples), i.e., it is the image of an anomaly. If one is interested in the physical reason which unites the two separated pieces of data given in the Moonshine, it is useful to search for this anomaly itself for this information. 

The goal of this paper is to show that 
\begin{proposition}\label{notDetectable}
The genus-one data applied to detect anomalies for the symmetry given by the dicyclic group of order $4N$, $\Dic_N$, has an undetectable $\bZ_2$ kernel. 
\end{proposition}

The structure of the paper is as follows: in section \ref{genus1section} we spell out the conditions that are specific to genus-one data, along with the construction of the stack $\mathcal{M}^G$. We also explain how to break down the question from a general finite group to studying $p$-groups.  In section \ref{partitionfunctionssection} we recast the method of detecting anomalies associated to a line bundle over $\cM^G$, to finding phases of 2d partition functions which are eigenvalues of acting with modular transformations. We investigate how genus-one constraints affect our ability to detect anomalies of dicyclic groups and show Proposition \ref{notDetectable}.  Section \ref{WZWexampleSection} contains an example where we apply the techniques of manipulating partition functions to see if we can fully detect the anomaly for $\widehat{\SU}(2)_k$ WZW model with quaternion symmetry. 
\section{Genus-One Data}\label{genus1section}
Consider a theory in (1+1)d which enjoys a global symmetry $G$.  We start with a stack~$\cM^G =(E, P)$ where $E$ is oriented and there exists an isomorphism $E \simeq \bT^2$. Furthermore, we equip $E$ with a $G$-bundle where $P: E \to \rB G$.  This stack has a standard presentation as follows: for any choice of isomorphism $f: E \to \bT^2$, the map $P\circ f$ is a $G$-bundle on the standard torus that has holonomies along the two cycles. We also choose a trivialization, $\varphi$, of $P\circ f$ at some basepoint that we will take to be the origin of $\bT^2$.  The stack $\cM^G$ is a quotient under the automorphisms of these two choices extra choices.  We can therefore write a stack $\widetilde{\mathcal{M}}^G$ that is a covering  stack of $\mathcal{M}^G$, more specifically, $\widetilde{\mathcal{M}}^G = \{E,P,f,\varphi\}$.  Once we have chosen $f$ then $E$ is no more data, so we are now talking about the space of bundles of the standard torus trivalized at the origin. This is the same as the set of maps 
\begin{equation}
     \hom(\pi_1 \bT^2|_{\text{origin}}\,,G) = \{(x,y) \in G \times G \,|\, [x,y]=1\},
\end{equation}
i.e. $\widetilde{\mathcal{M}}^G$ is the set of commuting pairs in $G$.
The map from $\widetilde{\mathcal{M}}^G \to {\mathcal{M}}^G$ presents $\mathcal{M}^G$ as a quotient groupoid of $\widetilde{\mathcal{M}}^G$ by forgetting the data of $f$ and $\varphi$. We note that $G$-bundles at a point are always trivalizable and there are $|G|$ many trivialization, so in order to forget $\varphi$ we quotient $\hom(\bZ^2\,,G)$ by ``changes of trivialization".  This gives $\hom(\bZ^2\,,G)// G$, where the $G$ action is by conjugation on the holonomies. In other words, the action on $g$ on  $(x,y)$
is given by 
\begin{equation}
    (x,y) \triangleleft g := (g^{-1} x g\,, g^{-1} y g )\,.
\end{equation}
To forget the data of $f$, we use the fact that any two isomorphism differ by an automorphism of the standard two-dimensional torus. We therefore also left-quotient $\hom(\bZ^2, G)$ by the group $\SL(2,\bZ )$. An element $\gamma \in \SL(2,\bZ)$ acts on $(x,y)^{\mathsf{T}}$, where $\mathsf{T}$ denotes the transpose of the row vector, by \begin{equation}
    \underset{\gamma}{\underbrace{\begin{pmatrix} a & b\\
    c & d
    \end{pmatrix}}} \triangleright (x,y)^{\mathsf{T}} = (x^a y^b, x^c y^d)^{\mathsf{T}}\,.
\end{equation}
\begin{rem}
We are using the fact that $x$ and $y$ commute so that the above formula gives an action.  The two actions by $\gamma$ and $g$ also commute with each other.  We see that as a groupoid
\begin{equation}
   \cM^G = \SL(2,\bZ ) \backslash\backslash \hom(\bZ^2, G) //G. 
\end{equation}
\end{rem}
Over each point in $\cM^G$ lives a torus bundle $\cE^G = \{E,P,z\in E\}$, where $z$ is a point in the torus $E$ and the fibers of the map to $\cM^G$ are oriented 2-tori which are the ``points" $E$ themselves in $\cM^G$. The map $P$ now takes $\cE^G \to \rB G$ by mapping $(E,P,z) \mapsto P(z)$.
If the theory has an anomaly $\alpha \in \rH^3(G; \rU(1))$ which maps $\rB G\to \rU(1)[3]$, then we can use the composed maps $P^* \alpha$ as a map from $\cE^G \to \rU(1)[3]$.  Here, the brackets denote the degree of suspension for the regular group $\rU(1)$.  Therefore, $\cM^G$ carries a line bundle which are the maps $(E,P) \mapsto \int_{E} P^* \alpha$; a line bundle over a groupoid is the same data as associating to every automorphism in the groupoid a $\rU(1)$ number.  In particular, a typical object of $\cM^G$ given by $(x,y)^{\sf{T}}$ and a typical automorphism of this object is given by $(\gamma, g)$ so that 
\begin{equation}\label{genusOneConstraint}
    \gamma \triangleright (x,y)^{\sf{T}}= (x,y)^{\sf{T}} \triangleleft
    g\,.
\end{equation}
Thus, in order to give the information about the line bundle, we need to assign for each point $(x,y)$ a group homomorphism, which is $\int_{E}P^* \alpha = \int \alpha: (\gamma,g) \to \rU(1)$. We do this in the following way. We start with a standard two-torus and wrap along the $a$ and $b$ cycle the elements $x$ and $y$, which attaches a $G$-bundle to this torus.  We now take the cylinder on the $G$-bundle, but apply a twist $\gamma$ to the two cycles.  Then, we take $g$ to change the trivialization of the $G$-bundle to return to a configuration that matches what we started with, and lastly identify the starting and ending tori. This procedure is depicted in Figure \ref{mappingCylinder}.  
\begin{figure}[t!]
    \centering
    \begin{tikzpicture}
    \draw[thick] (0,0) ellipse (1cm and 2cm);
    \draw[thick,domain=180:360,smooth,variable=\x] plot ({.2+.4*sin(\x)},{.9*cos(\x)});
\draw[thick,domain=23:157,smooth,variable=\x] plot ({-.1+.4*sin(\x)},{.9*cos(\x)});
\draw[thick] (0,2) -- (3.5,2);
\draw[thick] (0,-2) -- (3.5,-2);
\draw[thick] (3.5,2) -- (7,2);
\draw[thick] (3.5,-2) -- (7,-2);
\begin{scope}[rotate=-90,shift={(1.3,-.2)}]
\draw[thick,red,domain=10:130,smooth,variable=\x] plot ({.7*cos(\x)},{.4*sin(\x)});
\draw[thick,red,dashed,domain=190:300,smooth,variable=\x] plot ({.2+.7*cos(\x)},{.3+.4*sin(\x)});
\end{scope}
\draw[thick,blue] (0,.1) ellipse (.6cm and 1.4cm);
 \draw[thick] (3.5,0) ellipse (1cm and 2cm);
    \draw[thick,domain=180:360,smooth,variable=\x] plot ({3.5+.2+.4*sin(\x)},{.9*cos(\x)});
\draw[thick,domain=23:157,smooth,variable=\x] plot ({3.5+-.1+.4*sin(\x)},{.9*cos(\x)});
\draw[thick,blue] (3.5,.1) ellipse (.6cm and 1.4cm);
\begin{scope}[rotate=-90,shift={(1.3,-.2+3.5)}]
\draw[thick,red,domain=10:130,smooth,variable=\x] plot ({.7*cos(\x)},{.4*sin(\x)});
\draw[thick,red,dashed,domain=190:300,smooth,variable=\x] plot ({.2+.7*cos(\x)},{.3+.4*sin(\x)});
\end{scope}
 \draw[thick] (7,0) ellipse (1cm and 2cm);
    \draw[thick,domain=180:360,smooth,variable=\x] plot ({7+.2+.4*sin(\x)},{.9*cos(\x)});
\draw[thick,domain=23:157,smooth,variable=\x] plot ({7+-.1+.4*sin(\x)},{.9*cos(\x)});
\draw[thick,blue] (7,.1) ellipse (.6cm and 1.4cm);
\node at (0,-2.25) {\footnotesize$\color{red}{x}$};
\node at (.0,1.75) {\footnotesize$\color{blue}{y}$};
\node at (3.5,-2.25) {\footnotesize$\color{red}{\gamma x}$};
\node at (3.5+3.5,-2.25) {\footnotesize$\color{red}{\gamma x g}$};
\node at (3.5+3.5,1.75) {\footnotesize$\color{blue}{\gamma y g}$};
\node at (3.5,1.75) {\footnotesize$\color{blue}{\gamma y}$};
\node at (1.75,-2.5) {\scalebox{1.25}{$\color{black}{\gamma}$}};
\node at (1.75+3.5,-2.5) {\scalebox{1.25}{$\color{black}{g}$}};
\begin{scope}[rotate=-90,shift={(1.3,-.2+7)}]
\draw[thick,red,domain=10:130,smooth,variable=\x] plot ({.7*cos(\x)},{.4*sin(\x)});
\draw[thick,red,dashed,domain=190:300,smooth,variable=\x] plot ({.2+.7*cos(\x)},{.3+.4*sin(\x)});
\end{scope}
    \end{tikzpicture}
    \caption{Each two-torus has, wrapped along its cycles, a commuting pair of elements $x,y \in G$. In the third direction we draw the mapping cylinder first acting by $\gamma$ and then by $g$ between two-tori, with the ends identified.}
    \label{mappingCylinder}
\end{figure}
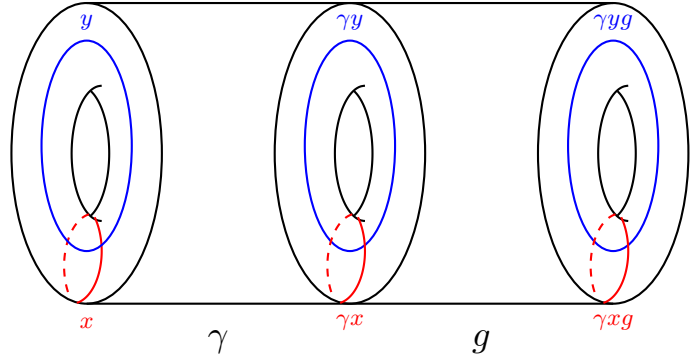 This gives a closed 3-manifold with a $G$-bundle that we can integrate $\alpha$ over. A form of this construction was given by \cite{ganter2009hecke} \footnote{While this reference constructs the analogue of $\cM^G$ with a line bundle, an action by any automorphism of the $G$-bundle on the line bundle
is given by multipication with the Chern-Simons invariant of the glued
mapping cylinder.}. 

The overall question can now be phrased in terms of the line bundle as follows: given~$\int \alpha \in \rH^1(\cM^G; \rU(1))$, with $\alpha$ an anomaly in $\rH^3(\rB G; \rU(1))$, then is it possible to determine the value of~$\alpha$?  The kernel of the map $ \rH^3(\rB G; \rU(1)) \overset{\int}{\rightarrow} \rH^1(\cM^G; \rU(1))$ is exactly our failure to be able to detect the anomaly. For any $G$, we can choose to restrict to a $p$-Sylow subgroup, i.e. a maximal $p$-group where every element is a power of $p$, denoted by $S$; we can do this prime by prime. 
It is therefore possible to restrict the cohomology $\rH^3(\rB G; \rU(1))$ along the $S$ subgroup,
\begin{center}
    \begin{tikzcd}
    {\rH^3(\rB G; \rU(1))_{(p)} } {\arrow[rr]} \arrow[d, hook,shift left=4] &  & {\rH^1(\cM^G; \rU(1))} \arrow[d] \\
{\rH^3(\rB S; \rU(1))_{(p)}} \arrow[rr]                                                                                &  & {\rH^1(\cM^S; \rU(1))}\,,                                        
\end{tikzcd}
\end{center}
where the subscript $p$ denotes $p$-local cohomology. The map from $\rH^3(\rB G; \rU(1))_{(p)}$ to~$\rH^3(\rB S; \rU(1))_{(p)}$ is a $p$-local injection \cite[\S XII.8]{CartanEilenberg}, but not an injection on the full cohomology, unless one takes a product over all $p$.  If there exists $G$ so that integration is not an injection, then there must be an $S$ such that integration is not an injection. To study this question on all groups we therefore focus on the $p$-groups.  

Given that $S$ is a $p$-group, we can use a fundamental fact of $p$-groups which states that for any $p$-group there exists a central order-$p$ element, thus we have $S = \bZ_p \,. S'$. To break down the problem even further, we can temporarily restrict $\alpha$ to the $\bZ_p$ subgroup, then by naturality we have $\int \alpha \big|_{\bZ_p}=\left(\int \alpha\right)\big|_{\cM^{\bZ_p}}$. We see that the central element does not contribute to the kernel by using the fact that:
\begin{lemma}\cite[\S3.3]{Gaberdiel:2012gf}
The map $\int \alpha|_{\bZ_p}: \rH^3(\rB \bZ_p; \rU(1)) \to \rH^1(\cM^{\bZ_p}; \rU(1))$ is injective, i.e., if $\int \alpha|_{\bZ_p}=1$ then $\alpha|_{\bZ_p}=1$.
\end{lemma}

\begin{rem}
The question can also be phrased in another form that is in terms of extensions rather than anomalies. For concreteness,
suppose that a theory has as its symmetry group, $G = \bZ_p \times G'$ where $G'$ is a finite $p$-group.
The only anomaly is mixed, living in~$\rH^2(G'; \rH^1(\bZ_p; \rU(1))) = \rH^2(G'; \widehat{\bZ}_p)$. If we gauge the $\bZ_p$ symmetry as in \cite{Bhardwaj:2017xup} we get a central extension $\bZ_p.G'$ symmetry action for the gauged theory, where the extension data is the mixed anomaly \cite{Tachikawa:2017gyf}. The question is therefore equivalent to asking: can one work out which extension using only genus-one data?
\end{rem}
From a categorical point of view ``modular data" of a modular tensor category means looking at its corresponding $\SL(2,\bZ)$ representation. The modular tensor category $\cZ(\textbf{Vec}^\alpha [G])$ has modular data, and it was shown in \cite{mignard2017modular} that is is insufficient to determine $\alpha$.  However, it was shown by Kirillov Jr. \cite{kirillov2002modular} that $\cZ(\textbf{Vec}^\alpha [G])$ along with the full data of the subcategory $\textbf{Rep}(G)$ was sufficient to determine $\alpha$.  Our current problem is an intermediate of these two situations. On the one hand we have more than modular data because we also incorporate data of the group that the modular tensor category came from, hence the fact that we can conjugation elements of $\cM^G$ by group elements.  On the other hand, we do not have the full category $\textbf{Rep}(G)$ to apply the Kirillov Jr. construction.

 \section{Partition Functions}\label{partitionfunctionssection}

Performing the integral over the mapping cylinder is in general hard to do and involves knowledge of how to triangulate the manifold, however, there are instances when this can be done. We can consider the case in which the mapping cylinder in Figure \ref{mappingCylinder} is $G$-equivariantly cobordant to the Lens space $L(N,1)$, or when the twists applied in the third direction is trivial, yielding a 3-torus $\bT^3$. This is the case when we are only concerned with $\bZ^k_N$ groups and the anomaly $\alpha \in \rH^3(\bZ^k_N; \rU(1))$.  The third cohomology evaluates to $\bZ^{\left[\binom{k}{1}+\binom{k}{2}+\binom{k}{3}\right]}_N$ and the cocyles are of the following three forms:
\begin{subequations}\label{3cycles}
\begin{align}
        \alpha^I(a,b,c) &= \exp\left( \frac{2\pi i q^I}{N^2} a^I(b^I+c^I-[b^I+c^I])  \right) \,, \\
    \alpha^{IJ}(a,b,c) &= \exp\left( \frac{2\pi i q^{IJ}}{N^2} a^I(b^J+c^J-[b^J+c^J])  \right)\,, \label{double}  \\
    \alpha^{IJK}(a,b,c) &= \exp \left(\frac{2\pi i q^{IJK}}{N} a^I b^J c^K  \right)\,\label{Triple},
\end{align}
\end{subequations}
where the superscript indices take values in $\{1,\ldots,k\}$, and $a,b,c \in \bZ^k_N$ \cite{deWildPropitius:1995cf}.  We denote $[b^I+c^I]: = b^I+c^I \mod N $, and $q^I, q^{IJ}, q^{IJK}$ takes values mod $N$, meant as a representative of the cocycle.  To argue why there are $\binom{k}{2}$ many cocycles of the form in \eqref{double} we note that
the 3-cocycles $\alpha^{I J}$ and $\alpha^{J I}$ are equivalent, since they differ by a coboundary.  A similar argument holds for cocycles of the third type in \eqref{Triple}, and therefore there are only $\binom{k}{3}$ many, as permutations of the labels $I,J,K$ give equivalent cocycles up to coboundary. These three types of cocycles correspond to the generators of $\rH^3(\bZ^k_N; \rU(1))$,
which at the level of gauge fields corresponds to self coupling of the gauge fields, pairwise couplings of the gauge fields, or coupling each of the three distinct fluxes together.
Each of these cocycles corresponds to a theory in (2+1)d and is the action of a {$G$-SPT}.  By the anomaly inflow mechanism, we can think of our (1+1)d theory with anomaly $\alpha$ as the boundary of this bulk $G$-SPT. While the boundary theory is anomalous, the entire bulk boundary set up is non-anomalous, thus the SPT exactly captures the anomaly data in its action.
The partition function for the $G$-SPT when placed on $L(N,1)$ is sufficient to detect the first two types of cocycles, while the last is detectable when placed on $\bT^3$  \cite{Tiwari:2017wqf}.  In particular, the partition function for each of the SPTs is a $\rU(1)$ valued  topological invariant used to distinguish the phase.  The partition functions are built out of a \textit{response function}, which treats the symmetry $G$ as a flat background connection; these functions can be shown to match the expression for
the group cocycles in  \eqref{3cycles}.  
Evaluating the partition function, i.e. integrating over $L(N,1)$, amounts to integrating the response function over a homology 1-cycle that generates $\rH_1(L(N,1),\bZ)$. 
The set of invaraints for the three cycles in \eqref{3cycles} is given by 
\begin{align}
    \bigg\{ \exp\left( \frac{2\pi i \,q_I}{N}  a^2_I \right)\,,
    \quad \exp \left( \frac{2\pi i \, q_{IJ}}{N}  a_{I} a_{J}\right) \,, \quad
    \exp\left( \frac{2\pi i q_{IJK}  }{N}  \epsilon^{ijk} a_{I,i} \,b_{J,j}\, c_{K,k} \right) \bigg \}\,,
\end{align}
were the indices $i,j,k$ on the last factor indicate the cycles on $\bT^3$.
We can also consider general discrete Abelian groups which are always isomorphic to  $\prod^k_{I=1} \bZ_{N^I}$; the SPTs can be detected on $L(N^I,1), L(\gcd(N^I,N^J)\,,1)$ and $\bT^3$.

We can convert the problem involving integrating over the mapping cylinder into the language of partition functions.  In (1+1)d, these are objects which transform as a modular form with respect to $\tau$ on the moduli space of flat 2-tori. If our theory enjoys a symmetry $G$, then the torus base manifold of our theory is equipped with a $G$-bundle and the map $P: \bT^2 \to \rB G $  is a pair of commuting elements (up to conjugation) $g, k \in G$ each wrapping one of the cycles of the torus. We define the partition function, with $q = \exp{2 \pi i \tau}$ and $\overline{q} =\exp{-2 \pi i \overline{\tau}} $, as 
\begin{equation}
    Z_{g,k}(\tau,\overline{\tau}) = \Tr_{\cH_k}\left(g\, q^{h-\frac{c}{24}} \,\overline{q}^{\bar{h}-\frac{\bar c}{24}}\right),   \qquad
    \begin{minipage} {2.0 cm}
    \begin{tikzpicture}
           \draw[thick] (0,0) -- (2,0)--(2,2)--(0,2)--cycle; 
           \draw[thick,red](0,1)--(2,1);
           \draw[thick,blue](1,0)--(1,.9);
           \draw[thick,blue](1,1.1)--(1,2);
           \draw (2.25,1) node {\textcolor{red}{$g$}};
           \draw (1,-.25) node {\textcolor{blue}{$k$}};
    \end{tikzpicture}
    \end{minipage}
    \label{fig:my_label}
\end{equation}
which is a configuration that is \textit{twisted} by $g$ in the spatial direction, and \textit{twined} by $k$ in the time direction. The trace is over the \textit{defect Hilbert space}, this is from the $k$-defect intersecting the spatial circle and implements a twisted
periodic boundary condition \cite{chang2019topological}.  These partition  functions are precisely the sections of the line bundle defined by  $\int \alpha$ over the stack $\cM^G$. An anomaly then has to do with an obstruction to this line bundle being trivializable.
For a special case where $G = \bZ_N$ we can consider the component $(g,e)$, where $e$ is the identity, of $\cM^G$.
A modular $S$ transformation on the partition function exchanges the two cycles of the torus so the $g$ defect now acts at a fixed time and the partition function is 
\begin{equation}
    Z_{e,g}(\tau,\overline{\tau})=\Tr_{\cH_g}\left(q^{h-\frac{c}{24}} \,\overline{q}^{\bar{h}-\frac{\bar c}{24}}\right) = Z_{g,e}\left(-\tfrac{1}{\tau}, -\tfrac{1}{\overline{\tau}}\right).
\end{equation}
Under the $T$ transformation, which maps $\tau \to \tau+1$, we see that this partition function is modular up to a multiplier of a phase which records the line bundle. To compute the phase we note that the spins $h-\overline{h}$ of the states in the defect Hilbert space takes value in $\frac{\ell}{N^2}+\frac{\bZ}{N}$, where $\ell$ is an integer modulo $N$ \cite{Lin:2021udi}, where it is referred to as a \textit{spin selection rule}. This implies the following, which was also mentioned in \cite{Gaberdiel:2012gf}:
\begin{proposition}
If an anomaly of the $G$ action is given by $\ell \in \rH^3(\bZ_N; \rU(1))$, $T^N = \begin{pmatrix}
1 & N \\
0 & 1
\end{pmatrix}$ acts on $Z_{e,g}$ with multiplier $\exp\left(\frac{2 \pi i \ell}{N} \right)$ .
\end{proposition}
 An immediate corollary is that knowledge of the partition function is sufficient to determine the anomaly for $\bZ_N$ groups. Going back to our picture using genus-one data and the mapping cylinder, this example for $\bZ_N$ groups would be what happens if we wrap $e,g$ along the cycles labeled by $x,y$ in Figure \ref{mappingCylinder} and
 apply $\gamma = T^N$ along the third direction giving the entire mapping cylinder the structure of a Lens space $L(N,1)$. Since the manifold used to detect the 3-cocycles for the case of a general discrete Abelian group is also a Lens space, or a 3-torus, then genus-one data is sufficient to detect anomalies of Abelian groups. Furthermore, it is sufficient to detect the anomaly for $S$ a $p$-group as in \S \ref{genus1section}, which has a  restriction to an Abelian $S'$. 

\begin{definition}\label{categoricalschur}
A subgroup $S \subseteq G$ is a \define{ categorical Schur detector  (CSD)} at $p$
if the restriction map $\rH^3(G; \rU(1)) \to \rH^3(S; \rU(1))$ on the $p$ parts is injective. More generally,
a set of subgroups $S \subseteq G$ is a \define{joint categorical Schur detector} at $p$ if the total restriction map $\rH^3(G; \rU(1)) \to \prod_{S}\rH^3(S; \rU(1))$ on the $p$ parts is injective.
\end{definition}
If the group $G$ has an Abelian joint CSD, i.e. one where all $S$ in the set are Abelian, then we would be able to detect the anomaly by our ability to integrate over Lens spaces for any Abelian group.
\begin{example}
The notion of CSD was also used in \cite{johnson2020h4} for cohomology in degree four where it was shown that for $G= \mathrm{Co}_0$, the linear isometry group of the Leech lattice, the restriction map $\rH^4(\mathrm{Co}_0; \bZ) \to \rH^4(S; \bZ)$ is injective, where $S$ is isomorphic to the product of the cyclic group
of order 3 and the binary dihedral or  group of order 16.
We see that a $p$-Sylow group of $G$ is also an example of a CSD at $p$ but the case in which $G$ is the extraspecial group $p^{1+2}_+$ does not have an individual CSD. The lack of a CSD comes simply from the fact that $\rH^3(G; \rU(1))$ has dimension 4 in $p$ while $\rH^3(S; \rU(1))$ only has dimension 3, so there is no injection. 
Take the case of $p=3$, it was shown in \cite{MinhEssential} that for $G$ an extraspecial $p$-group of order 27 with exponent 3 has no \textit{essential cohomology} in any degree. Essential cohomology is the $\bZ_p$-cohomology that gives the common kernel of the restrictions to all proper subgroups of $G$ as in definition \ref{categoricalschur}, i.e. the cohomology fits in the exact sequence 
\begin{equation}
    \rH^\bullet_{\text{Ess}}(G\,; \bZ_p) \to \rH^\bullet(G\,;\bZ_p) \to \prod_{S \subset G} \rH^\bullet(S\,; \bZ_p)\,.
\end{equation}
When we restrict to degree three, specifically with $\rU(1)$ coefficients by the standard long exact sequence, this measures the failure for there to be a joint CSD, so vanishing essential cohomology indicates there is a joint CSD in this case. 

An important and natural question is how to classify $p$-groups with non-zero
essential cohomology.  Let $G$ be an elementary Abelian $p$-group with rank $i>0$, the cohomology ring of $G$ is standard and given by 
\begin{align}
    \rH^\bullet(G\,; \bZ_p) = \begin{cases}
    \bZ_p[x_1,\,x_2\ldots,\,x_i] \hspace{36mm} p=2, \quad \text{deg}(x_i)=1 \\
    \bZ_p[x_1,\,x_2\ldots,\,x_i] \otimes \bigwedge (y_1,\,y_2,\ldots,\,y_i) \quad  p>2, \quad  2\,\text{deg}(y_i)=\text{deg}(x_i)=2\,.
    \end{cases}
\end{align}
For $p=2$, $\rH^\bullet(G\,; \bZ_p) \neq 0$, and for $p>0$ the essential cohomology is the Steenrod Closure of the product of $y_1 \cdots y_i$  \cite{Aksu2009EssentialCA}.
It was conjectured in \cite{CARLSON2001229} that the
essential cohomology of an arbitrary $p$-group is free and finitely generated over a certain polynomial subalgebra in $\rH^\bullet(G\,; \bZ_p)$; this conjecture holds for elementary $p$-groups.

\end{example}
Let us move to the case in which the global symmetry is a group that has even order by considering the dicyclic, or binary dihedral, group $\Dic_{N}$. A special case is $Q_8$ which is also an extraspecial group of order 8. We will show that:
\begin{proposition}\label{CSDforSU2}
Let $G$ be a subgroup of $\SU(2)$, then $G$ has no joint CSD.
\end{proposition}
  Recall that the dihedral group $\Dih_{N}$, a group of order $2N$ is the group of symmetries of a ${N}$-gon and lives as a subgroup of $\SO(3)$, where the reflection is implemented as a 180 degree rotation in 3d.  An 180 degree rotation lifts with order four to the double cover $\Spin(3)=\SU(2)$.  The restriction of $\Spin(3)$ along the map from $\rO(2) \to \SO(3)$ leads to the group $\Pin^-(2)$, where reflections square to $-1$.  The further restriction of $\Pin^-(2)$ along the map from $\Dih_{N} \to \rO(2)$ leads to $\Dic_{{N}}$;
the bindary dihedral groups are the ``discrete" versions of $\Pin^-(2)$. {This is summarized in the diagram below}
\begin{equation}
\begin{tikzcd}
\Dic_N \arrow[r] \arrow[dr, phantom, "\ulcorner", very near start] \arrow[d] & \Dih_{N} {} \arrow[d] \arrow[d, two heads] \\
\Pin^-(2){} \arrow[d] \arrow[dr, phantom, "\ulcorner", very near start] \arrow[r]  & \rO(2){} \arrow[d, two heads]           \\
\Spin(3){} \arrow[r]            & \SO(3)\,.{}                    
\end{tikzcd}
\end{equation}
\begin{proof}[Proof of Proposition \ref{CSDforSU2}]
The bindary dihedral group $G$ acts faithfully on $\Spin(3)$ which has the topology of a three-sphere. There is a fibration 
\begin{equation}
  \begin{tikzcd}
S^3/G \arrow[r] & S^3 \arrow[d] \\
                & \rB G           
\end{tikzcd}
\end{equation}
where $S^3/G$ is an oriented three-manifold, so has cohomology in degree three and below. From the fibration one can compute the group cohomology of $\rB G$. It is known that for any finite subgroup $G$ of the three-sphere that (see for example, \cite{epa2016platonic}):
 \begin{align}
    \rH^i(\rB G; \rU(1)) = \begin{cases}
    \rU(1) \quad  i=0\,,\\
    G^{ab} \quad i \equiv 1 \!\mod 4\,, \\
    \bZ_{|G|} \quad i \equiv 3 \!\mod 4\,, \\
    0 \quad \quad i > 0 \quad \! \text{and even}\,, \\
    \end{cases}
        \end{align}
where $G^{ab}$ denotes the abelianization of $G$ and $ \bZ_{|G|}$  denotes the group of complex $|G|$-th roots of unity.
If $S$ is a subgroup of $G$ then the restriction $\rH^3(G; \rU(1)) \to \rH^3(S; \rU(1))$ is a surjection and loses information about which subgroup of $\rH^3(G; \rU(1))$ it is, as it only dependents on the order of~$S$.
\end{proof}

We now see if genus-one data can still detect the anomaly. Let $G =  \Dic_N$ and consider wrapping the commuting pair of an element $g\in G$ and the identity $e$ around the cycles of the two-torus. As per Figure \ref{mappingCylinder}, we will let the group element which runs along the third direction of the mapping cylinder be $h$.  The elements $g$ (and $h$) could be reflections or rotations i.e. $g^{2N} = c$ or $g^{2}=c$ where $c$ is the central element. By \eqref{genusOneConstraint} it must be that 
\begin{equation}
    \gamma \begin{pmatrix} e \\
    g\end{pmatrix} = \begin{pmatrix}e\\ h\, g\, h^{-1} \end{pmatrix},
\end{equation}
and furthermore this is the most complicated configuration for the binary dihedral group that the constraints of genus-one will allow.  Take $g$ to be rotation, and $h$ to be a reflection, then $h g h^{-1}$ is $g^{-1}$.  So what are the possible $\gamma$'s? The second component of the vector after acting by $\gamma$ is $e^c g^d$ which must equal $g^{-1}$, thus $d = 2N-1$ and $c$ is free to be anything.  The first component is $e^a g^b = 1$. So $a$ is free but $g^b=1$, so $b=b'(2N)$.  In this case \begin{equation}
    \gamma = \begin{pmatrix}
    a &\,\quad b'(2N) \\
    c &\,\quad 2N-1
    \end{pmatrix}.
\end{equation}
We will take the matrix entries of $\gamma$  modulo $2N$ since the rotations is a cyclic group of order $2N$, and use the fact that $\det \gamma = 1$. But because $b$ is zero mod $2N$ the two valid matrices are
\begin{equation}
    \begin{pmatrix}
    -1 &\, 0 \\
    -c &\, -1
    \end{pmatrix}\,,\qquad  \begin{pmatrix}
    -1 &\, -2N \\
    0 &\, -1
    \end{pmatrix}\,,
\end{equation}
note that since $c$ was free to take any value mod $2N$, we write it as $-c$ in the matrix. We now take $h$ to be a rotation, and $g$ to be a reflection.  Then, in order for $hgh^{-1} = g^d$, it must be that $hg = g^d h$, which implies $g^{-1}h g = g^{d-1} h $ and so $h^{-1}= g^{d-1} h$.  But $h^{-2} = g^{d-1}$ has no solutions in general if $h$ is a generator of rotations.  For example, in the case of a 2-gon or 4-gon, it is possible to satisfy the equality. If $h$ and $g$ are both reflections then on the one hand $hgh^{-1}$ is given by taking $g$ and reflecting about the $h$ axis.  The value of $h g h^{-1}$ is $ g$ or $-g$ if $h$ and $g$ are the same reflection or off by 90 degrees, respectively. On the other hand, when acting by $\gamma$ we have $ e^c g^d = \pm g$ depending on whether $d$ is even or odd. 
 The case where $g=h$ is uninteresting, so the set of $\gamma$ is spanned by~$\begin{pmatrix}
 1& 0\\
 c& 1
\end{pmatrix}$.  If $h$ and $g$ were both rotations and thus cyclic subgroups, we know that restriction to any subgroup is not injective, so that will in general not be optimal for allowing us to detect the anomaly. 
We conclude that: 
\begin{proposition}
For the dicyclic group $\Dic_N$ where the anomaly is a $4N$-th root of unity, genus-one data contains the most information is when the whole group can be generated, i.e. when $g$ is a generator of rotation, and $h$ is a reflection. 
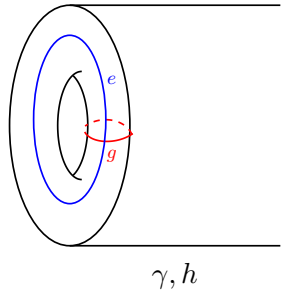
\begin{figure}[!h]
\adjustbox{scale=.8}{
    \centering
    \begin{tikzpicture}
    \draw[thick] (0,0) ellipse (1cm and 2cm);
    \draw[thick,domain=180:360,smooth,variable=\x] plot ({.2+.4*sin(\x)},{.9*cos(\x)});
\draw[thick,domain=23:157,smooth,variable=\x] plot ({-.1+.4*sin(\x)},{.9*cos(\x)});
\draw[thick] (0,2) -- (3.5,2);
\draw[thick] (0,-2) -- (3.5,-2);
\begin{scope}[rotate=-170,shift={(-.7,.1)}]
\draw[thick,red,domain=10:130,smooth,variable=\x] plot ({.48*cos(\x)},{.28*sin(\x)});
\draw[thick,red,dashed,domain=190:300,smooth,variable=\x] plot ({.2+.5*cos(\x)},{.3+.4*sin(\x)});
\end{scope}
\draw[thick,blue] (0,.1) ellipse (.6cm and 1.4cm);

\node at (.7,-.5) {\footnotesize$\color{red}{g}$};
\node at (.7,.75) {\footnotesize$\color{blue}{e}$};
\node at (1.75,-2.5) {\scalebox{1.25}{$\color{black}{\gamma}, h$}};
\end{tikzpicture}
}
\end{figure}

This forces $\gamma$ to be in the coset of the matrices
\begin{equation}\label{allgamma}
    \begin{pmatrix}
     -1 &\, 0\\
     -c &\, -1
     \end{pmatrix}\,, \qquad \begin{pmatrix}
     -1 &\, -2N \\
     0 &\, -1
     \end{pmatrix}. 
\end{equation}
\end{proposition}

Recall that when $G$ is just a cyclic group and $h$ is trivial, the choice of acting on the partition function by~$\begin{pmatrix}
     1 &\, |G| \\
     0 &\, 1
     \end{pmatrix}$
extracted a nontrivial $G$-th root of unity eigenvalue. Acting by the first matrix in \eqref{allgamma} shifts the modulus from $\tau \mapsto \frac{\tau}{\tau+c}$ and amounts to applying $T^c$ and then $S$ transformations to the partition function $ \Tr_{\cH}\left(g\, q^{h-\frac{c}{24}} \,\overline{q}^{\bar{h}-\frac{\bar c}{24}}\right)$, where $g$ is wrapped in the spatial direction.  By the spin selection rule, applying $T^c$ for $c \mod |G|$ will not produce a $|G|$-th root of unity.    We therefore expect that the second matrix in \eqref{allgamma} will detect the anomaly. We can test this on a theory which has as its symmetry a general dicyclic group, and defer the computation of the anomaly for a specific partition function and symmetry group to the next section. Let~$\cT_g = S\,\big[ \Tr_{\cH}\left(g\, q^{h-\frac{c}{24}} \,\overline{q}^{\bar{h}-\frac{\bar c}{24}}\right)\big]$, then acting by $\begin{pmatrix}
     -1 &\, -2N \\
     0 &\, -1
     \end{pmatrix}$ gives 
\begin{equation}
    \cT_g \xrightarrow{ -(T^{2N})} \exp\left(\frac{ \pi i \ell }{N}\right) \cT_{g^{-1}}\,,
\end{equation}
 where we have used that fact that $\begin{pmatrix}
     -1 &\, 0 \\
     0 &\, -1
     \end{pmatrix} \cT_g = \cT_{g^{-1}}$. But $\cT_{g^{-1}} = \cT_{hgh^{-1}}$, since $g$ is a rotation and $h$ is a reflection, and $\cT_{hgh^{-1}} = \cT_{g}$ by cyclicity of the trace.  At best we are able to detect only a $2N$-th root of unity. The map $ \rH^3(\rB G; \rU(1)) \overset{\int}{\to} \rH^1(\cM^G; \rU(1))$ therefore has a kernel that is at least of order 2. 
     
     The restriction to the cyclic subgroup of rotations gives $\{0,\,2N\}\!\! \mod 4N$ as the elements of the $\bZ_2$ kernel. One could hope to detect an anomaly $\alpha \in \{0,\,2N\}$. A common strategy when faced with anomalies and extensions of a group is to gauge some symmetry subgroup. The group $\Dic_{N} = C . \bZ_2$ as a nonsplit, noncentral extension with $C \cong \bZ_{2N}$ a normal subgroup. By the Serre spectral sequence we have $\rH^\bullet(\Dic_N\,; \rU(1)) \Leftarrow \rH^\bullet(\bZ_2\,;\rH^\bullet(C\,;\rU(1)))$, with the $E_2$ page:
     \begin{equation}
          E^{ij}_2 = \,\begin{array}{c|cccccccccc}
          j\\
          \\
          \Sym^2 \widehat{C} & \Sym^2 \widehat{C}& \\
          0 & 0&0& \ldots \\
           \widehat{C}&\bZ_2&\bZ_2&\bZ_2&\ldots\\
           \rU(1)&\rU(1)&\bZ_2&0&\bZ_2&0\\
          \hline 
           & 0 & 1 & 2 & 3 & 4 &  \quad i\,.
          \end{array}
     \end{equation}
     Where $\widehat{C}$ denotes the Pontryagin dual of $C$, and is the dual symmetry after gauging $C$. The entry $ \Sym^2 \widehat{C}$ survives on the $E_\infty$ page because it is the image of the restriction map $\rH^3(\Dic_N\,;\rU(1)) \twoheadrightarrow  \rH^3(C\,;\rU(1))$. The $\bZ_2$ in bidegree $(2,1)$ survives on $E_\infty$ for degree reasons; along with $\Sym^2 \widehat{C}$, these two contribute an order of already $4N$, and so the $d_2: E^{1,1}_2 \to E^{3,0}_2$ must be an isomorphism. The data of $\alpha$ living purely over $\bZ_2$ in $(2,1)$ now becomes the extension of the groups $\widehat{C}.\bZ_2$ for the gauged theory. In particular, this group is dihedral if $\alpha=0$ and again dicyclic if $\alpha = 2N$. Reflections lift with order 2 in former case, and order 4 in the latter. A reflection $h$ in the ungauged theory squares to $-1$ in the group, and lives on in $\bZ_2$ part of the gauged theory.  However, this is insufficient to tell if this $h$ is $-1$ in the $\bZ_2$ action, and thus in conclusion we are unable to distinguish the elements in the kernel.

\section{WZW Example}\label{WZWexampleSection}
In this section we present an example of attempting to detect the anomaly in a WZW theory with symmetry $G = Q_8$, and failing to fully capture all possible values of the anomaly.  The quaternion group not only fits the bill for Proposition \ref{CSDforSU2} but 
from the point of view of essential cohomology, it was shown in \cite{Adem1997EssentialCO} that a $p$-group has essential cohomology if all its elements of order $p$ are central. $Q_8$ is the unique group in which every element of order 2 is central.
We consider the WZW theory $\widehat{\SU}(2)_k$ which has ~$\frac{\SU(2)_{\rL}\times \SU(2)_{\rR}}{\bZ_2}$ symmetry (see \cite{cheng2020relative} for a summary of symmetries for WZW CFTs), and anomaly $(k,-k)$. We can consider the $\SU(2)_{\rL}$ symmetry, to which $Q_8 \subset \SU(2)_{\rL}$.  From computing $\rH^3(\rB Q_8\,; \rU(1))$, we know that this group should admit an anomaly that is mod 8 and therefore we take $k$ also mod 8. We deem that the anomaly is detectable if we can extract the full $\bZ_8$ group for the range of $k$. The generator $g$ of the $\bZ_4$ group of rotation is placed on one cycle of the torus, and the identity $e$ is placed on the other due to the fact that the pair must commute. The characters of $\widehat{\SU}(2)_k$ are given by the Weyl-Kac character formula and take the form \cite[Section 11]{Blumenhagen:2013fgp}:
\begin{equation}
    \chi^k_{\ell}(\tau,z) = \frac{\Theta_{\ell+1,k+2}(\tau,z)-\Theta_{-\ell-1,k+2}(\tau,z)}{\Theta_{1,2}(\tau,z)-\Theta_{-1,2}(\tau,z)}\,,
\end{equation}
with $0\leq \ell < k$ and the generalized $\SU(2)$ $\Theta$-functions defined as
\begin{equation}
    \Theta_{\ell,k}(\tau,z) = \sum_{n\in \bZ+\frac{\ell}{2k}}\, q^{kn^2}e^{-2\pi i n kz}\,.
\end{equation}
The partition function is defined by $Z(\tau,\overline{\tau},z,\overline{z}) = \sum_{j=1}^k \chi^k_j \,\overline{\chi}^k_j$. When twisted in the spatial direction by $g$, this gives 
\begin{equation}\label{twistedWZW}
    Z_{g,e} = \Tr_{\cH}\left(g\, q^{h-\frac{c}{24}} \,\overline{q}^{\bar{h}-\frac{\bar c}{24}}\,e^{-2\pi i z \hat{j}^3}\, e^{2\pi i \overline{z} \hat{j}^3}\right)\,,
\end{equation}
where $z$ is the chemical potential for the $\rU(1)$-charge and $\hat{j}^3$ plays the role of the operator which has as its eigenvalue the $\frac{\bZ}{2}$ representation of $\SU(2)$ in the usual angular momentum algebra. When $g$ wraps in the time direction, after applying an $S$ transformation to \eqref{twistedWZW}, we see that by conjugation we can take $g$ to $\hat{j}^3$ and thus giving
\begin{equation}
     \Tr_{\cH_{g}}\left( q^{h-\frac{c}{24}} \,\overline{q}^{\bar{h}-\frac{\bar c}{24}}\,e^{-2\pi i (z+\frac{1}{4}) \hat{j}^3}\, e^{2\pi i \overline{z} \hat{j}^3}\right)\,.
\end{equation}
This is because any element $\kappa\in Q_8$ can be written as $i \sigma_i$, for some Pauli matrix $\sigma_i$, in the Lie algebra of $\SU(2)$ and all $\SU(2)$ elements are conjugate to each other. When $g$ is applied in the spatial direction, unless $g$ is the central element, this breaks the global symmetry to $\rU(1)$ which is the centralizer of $g$. Any meaningful partition functions could then only have $\rU(1)$ elements wrapping the time direction, in particular, a $\rU(1)$ group spanned by $\exp\left(-2\pi i z \hat{j}^3 \right)$. Applying the $S$ transformation to \eqref{twistedWZW} so that $g$ is wrapped in the time direction essentially amounts to shifting~$z\mapsto z+\frac{1}{4}$.
Then applying $-(T^{2N})$, for $N=2$, on the partition function and using the spin selection rule gives a phase $\exp\left(2\pi i\, (2N)\left(\frac{\ell}{(2N)^2}+\frac{\bZ}{2N}\right) \right) = \exp\left(\frac{ \pi i \ell}{N} \right)$, which is only a fourth-root of unity.
\begin{rem}
There is another analogous computation we can conduct with free fermions.  The fact that the dicyclic group is a subgroup of $\SU(2)$ means it acts on $\bC^2$. It therefore also acts on the vertex algebra of two complex fermions.
The generators of rotations and reflections in this case are respectively 
\begin{equation}
    g \mapsto \begin{pmatrix}
     e^{ \pi i/N} & 0\\
     0 & e^{ -\pi i/N}
    \end{pmatrix}\,, \quad 
    h \mapsto \begin{pmatrix}
     0 & 1 \\
     -1& 0 
    \end{pmatrix}.
\end{equation}
We can compute twisted-twining genera for this vertex algebra, and just see how these modular functions transform under $\SL(2,\bZ)$. One subtlety to note is that this method uses a fermionic theory whose anomalies are classified differently than in the bosonic case. One should then take the appropriate restriction to the bosonic part of the full anomaly.
\end{rem}
\bibliography{Genus1}{}
\bibliographystyle{alpha}
\end{document}